\documentclass[10pt]{preprint}
\usepackage[british]{babel}
\usepackage{alltheorems}
\usepackage{modeltheory}
\usepackage{amsmath, amssymb}

\title{A note on clique-width and tree-width for structures}

\author{Hans Adler and Isolde Adler}

\date{1st June 2008}

\crc{F.4.1, G.2.2}

\newcommand{\mM}{\mathcal A}
\newcommand{\ar}{\operatorname{ar}}
\newcommand{\tw}{\operatorname{tw}}
\newcommand{\w}{\operatorname{w}}
\newcommand{\pw}{\operatorname{pw}}
\newcommand{\ucw}{\operatorname{ucw}}
\newcommand{\ucwf}{\operatorname{ucwf}}
\newcommand{\UCW}{\textup{UCW}}
\newcommand{\UCWF}{\textup{UCWF}}
\renewcommand{\card}[1]{\left|#1\right|}
\renewcommand{\restrict}{\upharpoonright}
\newcommand{\bigmid}{\;\big|\;}
\begin{document}

\maketitle

\begin{abstract}
  We give a simple proof that the straightforward generalisation of clique-width to arbitrary
  structures can be unbounded on structures of bounded tree-width.
  This can be corrected by allowing fusion of elements.
\end{abstract}

\section{Introduction}
Clique-width, introduced by Courcelle and Olariu, is a good measure for the complexity of a graph in the sense that many problems that are intractable in general become tractable when restricted to graphs of bounded clique-width. Moreover, clique-width of a graph is bounded by a function of tree-width, while the converse is not 
true \cite{CO2000}. However, the clique-width of a clique is two,
so trying to measure the complexity of a general structure by the clique-width of its Gaifman graph makes no sense. Therefore we are looking for a notion similar to clique-width of graphs, which should ideally have all of the following properties:

\begin{enumerate}
\item It is defined for arbitrary structures.
\item It specialises (essentially) to clique-width, in the case of
graphs.
\item It is bounded by a function of the tree-width of the Gaifman
graph.
\item Computationally hard problems should become tractable on
	instances of bounded width.
\item It does not increase on induced substructures.
\item The value mapping is an MS transduction.
\item Every set of structures that is the image of an MS transduction from trees
	has bounded width.
\item For fixed $k$ a decomposition of width $k$ (or width~$f(k)$) can be computed in
polynomial time.
\end{enumerate}

This paper mainly addresses the third criterion. 
We consider the generalisation of clique-width to structures as 
proposed by Grohe and Tur\'an~\cite{GH2004}.
(A similar notion was introduced by Fischer and Makowsky, 
resulting in values that are smaller by one because elements can also be 
uncoloured~\cite{FM2004}.)
We give a simpler proof for a result of Courcelle, Engelfriet and
Rozenberg~\cite[Theorem 7.5]{CER1993}, showing that the clique-width of a class of
structures is not bounded in terms of its tree-width.
More specifically we show that there is a class of structures, with one ternary
relation symbol,
that has unbounded clique-width while the tree-width of the class is bounded
by~$2$. It is known that there can be no such examples with structures that have only
binary and unary relation symbols \cite{CO2000, Scheffzik}.

We then consider a modified definition where we additionaly allow fusion of elements.
We show that the corresponding modified width is bounded from above by 
tree-width plus~$2$. In an earlier version of this paper we referred to
clique-width with fusion as `reduced clique-width', since we were
not aware that the fusion operation had already been introduced by
Courcelle and Makowsky~\cite{CM2002}. We have now corrected this to avoid
unnecessary proliferation of technical terms.

\section{Structures, decompositions, and unary clique-width}
%We will only deal with finite 1-sorted relational structures:
A \emph{signature} $\sigma=\{R_1,\ldots, R_n\}$ 
is a finite set of \emph{relation symbols} $R_i$, $1\le i\le n$.
As usual, every relation symbol $R\in\sigma$ has an associated \emph{arity} $\ar(R)$.
A \emph{$\sigma$-structure} is a tuple 
$\mM=(A,R^{\mM}_1,\ldots, R^{\mM}_n)$ where $A$ is
a set, the \emph{universe} of $\mM$, 
and $R^{\mM}_i\subseteq A^{\ar(R_i)}$ for $1\leq i\leq n$.
All structures in this paper are finite (i.e.\ they have finite universes).

Given a structure $\mathcal A$, we write $G_{\mathcal A}$ for the
\emph{underlying graph} (also called \emph{Gaifman graph}) of 
$\mathcal A$: The vertices are
the elements of the universe, and two different vertices are
joined by an edge if, and only if, they appear together in
some tuple that is in a relation of $\mathcal A$.

A \emph{tree decomposition} of a graph $G=(V,E)$ is a pair 
$(T,B)$, consisting of a rooted
tree 
$T$ and a family $B=(B_t)_{t \in T}$ of subsets of $V$, the 
\emph{pieces} of $T$, satisfying:
	\begin{itemize}

	\item
	For each $v \in V$ there exists $t \in T$, such that $v \in B_t$.

	\item
	For each edge $e \in E$ there exists $t \in T$, such that $e \subseteq B_t$.

	\item For each $v \in V$ the set
	$\{t \in T \mid v \in B_t \}$ is connected in $T$.
\end{itemize}

\medskip \noindent The \emph{width} of $(T,B)$ is defined as
$\w(T,B):= 
\max\big\{\left|B_t\right|-1\ \bigmid t\in T\big\}.$
The \emph{tree-width of $G$} is defined as
$
	\tw(G):= 
	\min\big\{\w(T,B)\ \bigmid (T,B) \text{ is a tree decomposition of }G\big\}.
$

\emph{Path decompositions} and \emph{path-width} of $G$, $\pw(G)$, are 
defined analogously, 
with the additional restriction that $T$ be a path.

%For a graph $G$, $\tw(G)$ denotes the tree-width of $G$ and $\pw(G)$ its path-width.

%One notion of clique-width of structures was defined by Grohe and Tur\'an \cite{GH2004}. 
%A variant of this definition, used by Fischer and Makowsky, results in values that are 
%smaller by one because elements can also be uncoloured \cite{FM2004}.

%Unary clique-width was defined by Grohe and Tur�n\marginpar{\small G+T correct? reference!}. Till Scheffzik\marginpar{\small TS reference} proved that for a fixed(?) signature whose relation symbols are at most binary the unary clique-width of a structure is bounded by a function of its tree-width. We will now show that for general signatures the unary clique-width is bounded by a function of the path-width, and that this result is optimal even for a single ternary relation.

Let us call a pair $(\mathcal A,\gamma)$ consisting of a structure $\mathcal A$ and a map $\gamma: A\to\omega$ a \emph{coloured structure}. It is $k$-\emph{coloured} if $\gamma(A)\subseteq\{0,1,\dots,k-1\}$. We will usually think of elements of
colour 0 as `uncoloured', and identify structures and 1-coloured structures.
We call an element $a$ of a coloured structure $(\mathcal A,\gamma)$
\emph{isolated} if $a$ is the only element of colour $\gamma(a)$.
For a signature $\sigma$ and a non-negative integer $k$, we define $\UCW_k[\sigma]$ as the smallest class of $k$-coloured $\sigma$-structures such that:

\begin{enumerate}
\item Every ($1$-coloured) empty $\sigma$-structure is in $\UCW_k[\sigma]$.
\item Every ($1$-coloured) $1$-element $\sigma$-structure whose relations are all empty is in $\UCW_k[\sigma]$.
\item The disjoint union $(\mathcal A\sqcup\mathcal B,\gamma_A\sqcup \gamma_B)$ of two coloured structures $(\mathcal A,\gamma_A),(\mathcal B,\gamma_B)\in\UCW_k[\sigma]$ is again in $\UCW_k[\sigma]$.
\item If $(\mathcal A,\gamma)\in\UCW_k[\sigma]$ and $f:\{0,\dots,k-1\}\to\{0,\dots,k-1\}$ is any function, then $(\mathcal A,f\circ \gamma)\in\UCW_k[\sigma]$.
\item If $(\mathcal A,\gamma)\in\UCW_k[\sigma]$, $R\in\sigma$ is an $n$-ary relation symbol, and $c_0,\dots,c_{n-1}\in \{0,\dots,k-1\}$ is an $n$-tuple of colours, then for the structure $\mathcal B$ which is like $\mathcal A$ except that $\mathcal B\models R(a_0,\dots,a_{n-1})$ holds iff $\mathcal A\models R(a_0,\dots,a_{n-1})$ or $\gamma(a_i)=c_i$ for $i=0,\dots,n-1$, we have $(\mathcal B,\gamma)\in\UCW_k[\sigma]$.
\end{enumerate}

For a $\sigma$-structure $\mathcal A$ (possibly coloured) let $\ucw \mathcal A$, the \emph{unary clique-width of} 
$\mathcal A$, be the smallest number $k$ such that $\mathcal A\in\UCW_k[\sigma]$. `Unary' because only single 
elements are coloured.
It is easy to see that our notion is equivalent to that of Grohe and 
Tur\'an, even though they use a more restrictive colouring operation \cite{GH2004}.

Every $k$-coloured $\sigma$-structure $\mathcal A$ with $\card A\leq k$
elements has unary clique-width $\ucw\mathcal A\leq k$:
Take the disjoint union of $\card A$
differently coloured 1-element structures, introduce the necessary relations,
and recolour all elements with their final colour.

\section{Unary clique-width and tree-width}

\begin{proposition}\label{PropositionUcwLessThanPw}
	Every (1-coloured) structure $\mathcal A$ satisfies
	\[ \ucw(\mathcal A)\leq\pw(G_{\mathcal A})+2. \]
\end{proposition}

\begin{proof}
	Let $k=\pw(G_{\mathcal A})+1$.
	We will show by induction on $n$ the stronger statement that if $G_{\mathcal A}$
	has a path decomposition $B_0,B_1,\dots,B_{n-1}$ of width $\leq k+1$
  and $\gamma\colon A\rightarrow\{0,\dots,k\}$
	is a $k+1$-colouring of $\mathcal A$ such that all elements of $A\setminus B_{n-1}$
	have colour 0, then $(\mathcal A,\gamma)\in\UCW_{k+1}[\sigma]$.
	The case $n=0$ is trivial because then $\mathcal A$ is the empty $\sigma$-structure.
	
	For $n\geq1$, let $\mathcal A'=\mathcal A\restrict(B_0\cup\dots\cup B_{n-2})$
	be the induced substructure	of $\mathcal A$ with domain $B_0\cup\dots\cup B_{n-2}$,
	and let $\gamma'\colon B_0\cup\dots\cup B_{n-2}\rightarrow\{0,\dots,k+1\}$
	be a colouring such that every element of $B_{n-2}\cap B_{n-1}$
	is isolated and has a non-zero colour, while all other elements have colour~0.
	Then $(\mathcal A',\gamma')\in\UCW_{k+1}[\sigma]$ by the induction hypothesis.
	Let $\mathcal A''$ be the $\sigma$-structure with domain $A$ whose relations
	are precisely those of $\mathcal A'$.
	Let $\gamma''\colon A\rightarrow\{0,\dots,k+1\}$ extend $\gamma'$ so that
	any two elements of $B_{n-1}$ have distinct non-zero colours.
	$(\mathcal A'',\gamma'')$ can be obtained from $(\mathcal A',\gamma')$ by
	disjoint union with one-element structures, so clearly
	$(\mathcal A'',\gamma'')\in\UCW_{k+1}[\sigma]$. All relations of $\mathcal A$
	that are not also relations of $\mathcal A'$ must be between elements of $B_{n-1}$.
	Since every element of $B_{n-1}$ has a unique colour we can introduce all these
	relations without introducing any unwanted relations.
	Similarly, for $f\colon\{0,\dots,k+1\}\rightarrow\{0,\dots,k+1\}$ such that
	$f(0)=0$ and $f(\gamma''(a))=\gamma(a)$ we get $\gamma=f\circ\gamma''$.
	Hence $(\mathcal A,\gamma)\in\UCW_{k+1}[\sigma]$.
\end{proof}

\begin{lemma}
	For every structure $\mathcal A$ in the signature $\sigma=\{E\}$
	of graphs, there is a structure $\mathcal A'$ with
	universe $A'=A\sqcup\{t\}$ in the signature $\sigma'=\{R\}$, $R$ a
	ternary relation symbol, which satisfies:
	\begin{enumerate}
	\item $\tw(G_{\mathcal A'})=\tw(G_{\mathcal A})+1$.
	\item $\pw(G_{\mathcal A'})=\pw(G_{\mathcal A})+1$.
	\item $\ucw(\mathcal A')\geq\pw(G_{\mathcal A'})+1$.
	\end{enumerate}
\end{lemma}

\noindent (Note: The $+1$ in clauses 1 and 2 is due to the fact that 
$G_{\mathcal A'}$ is an apex graph over
$G_{\mathcal A}$, whereas the $+1$ in 3 merely corrects the 
conventional $-1$ in the definition of path-width.)

\begin{proof}
	We interpret the relation $R$ as follows:
	\[
		\mathcal A'\models Racb\quad\iff\quad a\not=b\text{ and } c\in\{a,b\} \text{ and }
			\big(\mathcal A\models Eab\text{ or }a=t\big).
	\]
	Then 1 and 2 hold because $G_{\mathcal A'}$ is an apex graph
	over $G_{\mathcal A}$.
	(I.e., $G_{\mathcal A}$ is an induced subgraph of $G_{\mathcal A'}$,
	and $G_{\mathcal A'}$ has exactly one additional vertex, say $t$, and $t$ has an edge
	with every vertex of $G_{\mathcal A}$.)
	Towards a proof of 3 we observe:
	
	If $\mathcal A'\in\UCW_k[\sigma']$, then there is a tree of $k$-coloured
	$\sigma'$-structures such that the leaves are singletons with empty relations,
	and every inner node is either binary and the disjoint union of its two children,
	or unary and obtained from its only child by recolouring or by introducing
  a new relation.
	Due to the definition of $R$, for every element $a\neq t$ there must be a node
	containing both $a$ and $t$ as isolated elements,
	and for any two distinct elements $a,b\neq t$ there must be a node containing
	$a,b,t$ as three isolated elements.

	The branch $(\mathcal A^t_0,\gamma^t_0),\dots,(\mathcal A^t_{m-1},\gamma^t_{m-1})$
	of the tree which begins with the root and ends in the leaf consisting of the single
	element $t$ has the following properties:
	\begin{enumerate}
	\item For every element $a\in A$ the branch contains a node that has $a$ and $t$
	as isolated elements.
	\item For every edge $\mathcal A\models Eab$ of $\mathcal A$
	there is a node that has $a,b,t$ as isolated elements.
	\item For every element $a\in A$ there is a greatest index $j$ such that
  $a\in A^t_j$, and a smallest index $i\leq j$ such that $a$ is the only element of
  colour $\gamma^t_i(a)$.
	\end{enumerate}
	Let $\ell$ be the greatest index such that $t$ is isolated in $\mathcal A^t_\ell$.
	For $i=0,\dots,\ell$ let $B_i$ consist of the isolated elements
	of $(\mathcal A^t_i,\gamma^t_i)$.
	Then $B_0\setminus\{t\},\dots,B_\ell\setminus\{t\}$ is clearly a path decomposition
	of $\mathcal A$. Each bag $B_i\setminus\{t\}$ has at most $k$ elements,
	so the width of the path decomposition is at most $k-1$.
	Thus we have shown that $\ucw\mathcal A'\leq k$ implies $\pw\mathcal A\leq k-2$.
\end{proof}

\begin{corollary}\label{CorollaryUnaryProblem}
	For every non-negative integer $n$ there is a structure $\mathcal A$ with only one,
	ternary, relation symbol, such that
	\begin{enumerate}
		\item $\tw(G_{\mathcal A})=2$ and
		\item $\ucw(\mathcal A)>n$.
	\end{enumerate}
\end{corollary}

\begin{proof}
	Let $\mathcal G$ be an undirected tree in the signature $\{E\}$
	of graphs, satisfying $\pw(G_{\mathcal G})\geq n$.
	Let $\mathcal A'=\mathcal G'$ be as in the Lemma.
	Then $\tw(G_{\mathcal G'})=\tw(G_{\mathcal G})+1=2$ and
	$\ucw(\mathcal G')\geq\pw(G_{\mathcal G'})+1=\pw(G_{\mathcal G})+2>n$.
\end{proof}

As Johann Makowsky pointed out to us, a result
of Glikson and Makowsky \cite{GM2003} is cited incorrectly
in the paper by Fischer and Makowsky \cite{FM2004} (as Theorem 3.7 (ii)),
resulting in an apparent contradiction to Corollary \ref{CorollaryUnaryProblem}.
Our result is optimal in the sense that, as shown by
Courcelle and Olariu, clique-width of structures with at most binary relations
is bounded by a function of tree-width of the Gaifman graph~\cite{CO2000}.
(The details for an arbitrary number of binary relations were checked
by Till Scheffzik in his diplomarbeit~\cite{Scheffzik}.)

\section{Clique-width with fusion}

The reason why the proof of Proposition~\ref{PropositionUcwLessThanPw} works with path-width but not with tree-width is that for unary clique-width there is no way to glue together two substructures that intersect in a small bag of elements. In other words, while we can introduce new relations in the signature depending only on the colours of the elements, we cannot do this for equality, which can also be seen as a binary relation. The definition below fixes this, using the fusion operation
of Courcelle and Makowsky \cite{CM2002,Makowsky2004}.

We will need a rarely used universal construction for structures: The quotient by an equivalence relation which need not be a congruence relation. Let $\mathcal A$ be a
$\sigma$-structure, and let $\sim$ be an equivalence relation on its domain $A$.
Let $B=A/{\sim}$ be the set of equivalence classes of $\sim$. Let $\mathcal B$ be
the $\sigma$-structure with domain $B$ which satisfies $\mathcal B\models R(b_1,\dots,b_{n(R)})$
if and only if there are $a_1\in b_1,\dots,a_{n(R)}\in b_{n(R)}$ such that
$\mathcal A\models R(a_1,\dots,a_{n(R)})$. We denote this \emph{quotient structure} by
$\mathcal B=\mathcal A/{\sim}$. ($\mathcal B$ is universal in the sense that
the projection map $p: \mathcal A\to \mathcal B$ is a homomorphism, and for every other
structure $\mathcal B'$ with the same domain $B$, if the projection map
$p': \mathcal A\to\mathcal B'$ is also a homomorphism then it factors through $p$.)
For a coloured structure $(\mathcal A,\gamma)$ and a set $C\subseteq\omega$ of colours
we define $a\sim_Ca'$ to mean $a=a'$ or $\gamma(a)=\gamma(a')\in C$. The obvious colouring
$\gamma/{\sim_C}$ induced on the quotient structure $\mathcal A/{\sim_C}$ satisfies
$\gamma/{\sim_C}(b)=\gamma(a)$ for some/any $a\in b$.

For a signature $\sigma$ and a non-negative integer $k$, the class $\UCWF_k[\sigma]$ is the smallest class of $k$-coloured $\sigma$-structures such that:

\begin{enumerate}
\item Every empty $\sigma$-structure is in $\UCWF_k[\sigma]$.
\item Every $1$-element $\sigma$-structure whose relations are all empty is in $\UCWF_k[\sigma]$.
\item The disjoint union $(\mathcal A\sqcup\mathcal B,\gamma_A\sqcup \gamma_B)$ of two coloured structures $(\mathcal A,\gamma_A),(\mathcal B,\gamma_B)\in\UCWF_\sigma^k$ is again in $\UCWF_k[\sigma]$.
\item If $(\mathcal A,\gamma)\in\UCWF_k[\sigma]$ and $f:\{0,\dots,k-1\}\to\{0,\dots,k-1\}$ is any function, then $(\mathcal A,f\circ \gamma)\in\UCWF_k[\sigma]$.
\item If $(\mathcal A,\gamma)\in\UCWF_k[\sigma]$, $R\in\sigma$ is an $n$-ary relation symbol, and $c_0,\dots,c_{n-1}\in \{0,\dots,k-1\}$ is an $n$-tuple of colours, then for the structure $\mathcal B$ which is like $\mathcal A$ except that $\mathcal B\models R(a_0,\dots,a_{n-1})$ holds iff $\mathcal A\models R(a_0,\dots,a_{n-1})$ or $\gamma(a_i)=c_i$ for $i=0,\dots,n-1$, we have $(\mathcal B,\gamma)\in\UCWF_k[\sigma]$.
\item If $(\mathcal A,\gamma_A)\in\UCWF_\sigma^k$ and $c\in\{0,\dots,k-1\}$, then also
$(\mathcal A/{\sim_{\{c\}}},\gamma/{\sim_{\{c\}}})\in\UCWF_k[\sigma]$.
\end{enumerate}

The only difference to the definition of $\UCW_k[\sigma]$ is that we
have added the fusion operation at the end.
For a $\sigma$-structure $\mathcal A$ let $\ucwf \mathcal A$, the \emph{unary clique-width with fusions} of $\mathcal A$, be the smallest number $k$ such that $\mathcal A\in\UCWF_k[\sigma]$.

\begin{remark}
	Every $k$-coloured $\sigma$-structure $(\mathcal A,\gamma)$ with at most $k$
	elements satisfies $(\mathcal A,\gamma)\in\UCWF_k[\sigma]$.
\end{remark}

\begin{proof}
	Clearly for every $\sigma$-structure $\mathcal A'$ with no relations and every
	$k$-colouring $\gamma'$ of $\mathcal A'$ we have
  $(\mathcal A',\gamma')\in\UCWF_k[\sigma]$.
	Now for $\mathcal A$ as in the remark let $\mathcal A'$ be the
  $\sigma$-structure which
	has the same domain as $\mathcal A$ but no relations, and let $\gamma'$ be a colouring
	of $\mathcal A'$ such that every element has a different colour.
	Since $(\mathcal A',\gamma')\in\UCWF_k[\sigma]$, clearly also
	$(\mathcal A,\gamma)\in\UCWF_k[\sigma]$.
\end{proof}

\begin{lemma}
	Suppose $\mathcal A$ has a tree decomposition $(T,B)$ of width $\leq k$,
	and $B_t$ is the bag at some tree node $t\in T$. Let $\gamma$ be a $(k+2)$-colouring
	of $\mathcal A$ such that all elements $a\in A$ of the domain of $\mathcal A$
	have colour $\gamma(a)=0$.
	Then $(\mathcal A,\gamma)\in\UCWF_k[\sigma]$.
\end{lemma}

\begin{proof}
	By induction on the minimal number of nodes in a tree decomposition of $\mathcal A$
	in which the bag $B_t$ occurs, using the remark as induction base.
\end{proof}

\begin{corollary}\label{PropositionUrcwLessThanTw}
	Every relational structure $\mathcal A$ satisfies
	\[ \ucwf(\mathcal A)\leq\tw(\mathcal A)+2. \]
\end{corollary}
\noindent The term $+2$ consists of $+1$ correcting for the conventional $-1$ in the
definition of tree-width, and $+1$ for the colour $0$, which we have put to
special use.

\section{Conclusion}

We have shown that clique-width for structures
fails even the basic requirement that it be bounded in terms of tree-width.
We have also shown that this can be corrected by admitting the fusion operation
in addition to the other operations, although we have not explored the
(presumably negative) impact of this on the other desirable conditions.

Another approach would be to colour tuples rather than elements.
Examples for definitions that go in that direction are Blumensath's
partition-width \cite{Blumensath2006}, and patch-width, defined by Fischer
and Makowsky and examined further by Shelah and Doron~\cite{FM2004,DS2007}.
For instance, we could assign colours to tuples of length at most
some fixed $n$. But then it
should still be possible to adapt our example to show that
we do not obtain a bound in terms of tree-width.

The authors thank Bruno Courcelle and Johann Makowsky for valuable suggestions.
%, and they have forgotten
%to replace this last statement with a final one listing more specific people,
%except .

%The authors thank Bruno Courcelle for helpful suggestions.

\end{document}